\documentclass[aps,pra,reprint, superscriptaddress]{revtex4-1}
\usepackage[utf8]{inputenc}

\usepackage{epsfig,graphics,graphicx}
\usepackage{dcolumn}
\usepackage{bm}
\usepackage{color}
\usepackage{amsmath,amssymb,amsthm}
\usepackage{float}
\usepackage{centernot}
\usepackage{amsfonts}
\usepackage{mathtools}
\usepackage{fullpage}
\usepackage{units}
\usepackage[usenames,dvipsnames]{xcolor}
\usepackage{braket}
\usepackage{indentfirst}

\usepackage[T1]{fontenc}

\usepackage{mathpazo}
\usepackage{dsfont}

\usepackage{mathptmx}
\usepackage{latexsym}

\usepackage[colorlinks]{hyperref}
\hypersetup{pdfpagemode=UseNone, allcolors=blue}

\newtheorem{theorem}{Theorem}
\newtheorem{definition}[theorem]{Definition}

\newtheorem{lemma}[theorem]{Lemma}
\newtheorem{proposition}[theorem]{Proposition}

\newcounter{rem}
\def\>{\rangle}
\def\<{\langle}

\newcommand{\borb}[2]{\left | #1 \rangle\!\langle #2 \right |}
\renewcommand{\rho}{\varrho}

\def\ii{{\rm i}}

\def\textbf#1{{\bf #1}}
\newcommand{\Cx}{\mathbb{C}}
\newcommand{\Ir}{\mathbb{Z}}
\newcommand{\Nl}{\mathbb{N}}

\newcommand{\Rl}{\mathbb{R}}

\def\beq{\begin{equation}}
\def\eeq{\end{equation}}
\def\beqa{\begin{eqnarray}}
\def\eeqa{\end{eqnarray}}
\def\eea{\end{array}}
\def\bea{\begin{array}}
\newcommand{\bei}{\begin{itemize}}
	\newcommand{\eei}{\end{itemize}}
\newcommand{\bee}{\begin{enumerate}}
	\newcommand{\eee}{\end{enumerate}}
\def\bep{\begin{proposition}}
	\def\eep{\end{proposition}}
\def\bel{\begin{lemma}}
	\def\eel{\end{lemma}}
\def\bet{\begin{theorem}}
	\def\eet{\end{theorem}}
\def\bed{\begin{definition}}
	\def\eed{\end{definition}}

\begin{document}

\title{Quantum-to-classical transition via quantum cellular automata}

\author{Pedro C.S. Costa}
\email{pedro.costa@mq.edu.au}
\affiliation{Department of Physics and Astronomy,
	Macquarie University, Sydney, NSW 2109, Australia}

\date{\today}
\begin{abstract}
A quantum cellular automaton (QCA) is an abstract model consisting of an array of finite-dimensional quantum systems that evolves in discrete time by local unitary operations. Here we propose a simple coarse-graining map, where the spatial structure of the QCA is merged into effective ones. Starting with a QCA that simulates the Dirac equation, we apply this coarse-graining map recursively until we get its effective dynamics in the semiclassical limit, which can be described by a classical cellular automaton. We show that the emergent-effective result of the former microscopic discrete model converges to the diffusion equation and to a classical transport equation under a specific initial condition. Therefore, QCA is a good model to validate the quantum-to-classical transition. 
\end{abstract}
\maketitle

\section{Introduction}

In a realistic scenario, quantum systems are never completely isolated from the environment and its quantum effects naturally start to become suppressed. Such a phenomenon is known as the decoherence process, and it is
the main ingredient used to explore the quantum-to-classical transition~\cite{schlosshauer2014quantum}. Completely positive trace-preserving (CPTP) maps are powerful tools to study such transition. These maps, which can be seen as a unitary interaction among three systems and further discard two of them~\cite{wolf2012quantum}, are therefore not reversible and their application implies a loss of information about the system. Similarly with the classical picture, where coarse-graining (CG) procedure via the projective operators~\cite{PhysicsPhysiqueFizika.2.263,ehrenfest1990conceptual} helps us to extract effective dynamics that have lower computational cost, the CPTP maps also provide a data compressor for the quantum system. Thus, these maps when applied in quantum systems can also be seen as a coarse-graining procedure.

Rather than verify the necessary conditions to extract the quantum channels at each coarse-grained level~\cite{CG_QM,duarte2020investigating}, here we concentrate on just one specific quantum dynamics, one map and one CG level. The Dirac equation is the quantum dynamics chosen, and the coarse-graining level is the one where the coherence disappears. Given that the quantum model of computation, known as quantum cellular automata (proposed in~\cite{QCAviaQws}) is able to simulate a large class of quantum walks, in particular the one that simulates the Dirac equation, we use it to study the semiclassical limit of the Dirac equation~\cite{Dirac}.  Our procedure respects the same constraints between different CG levels for the time evolution shown in \cite{CG_QM,Oleg}, which is represented in figure (\ref{fig:CG}). However, our formalism takes advantage of the space structure provided by the CA, likewise done in \cite{CG_CA,costaPCA}. That is, in order to analyze the emergent dynamics in the semi-classical limit, we assume that we cannot resolve neighboring cells and then some information has to be discarded. 

Although in this work is the indistinguishability between the states the source of decoherence to study the quantum-to-classical transition dynamics, the procedure here, in analogy to \cite{chen2019quantifying},  can also be helpful to track the nonclassicality that emerges due to other sources of decoherence in the new quantum technologies.

The paper is organized as follows: In section \ref{sec:PUQCA} we defined the quantum model of computation used to simulate the Dirac equation in the continuous limit. In section \ref{sec:CG map} we introduce the coarse-graining map used to study the decoherence process for the massless Dirac equation. In section \ref{sec:CGQCA} the CG map is applied and our main results are established. Finally, in section \ref{sec:conclusion}, we conclude with a discussion of these results and we point out possible directions to go for next.

\section{A PUQCA that simulates the  (1+1) Dirac equation}\label{sec:PUQCA}

In one spatial and one temporal dimension, the Dirac equation of a free spin-$1/2$ particle of mass $m$ may be written as follows:
\begin{equation}
  \left(\ii \frac{\hbar}{c} \partial_t + \ii\hbar  \sigma_z \partial_x - m c \sigma_x \right)\psi=0,
  \label{eq:Dirac}
\end{equation}
where $c$ is the velocity of light, $\hbar$ is Planck's constant divided by $2\pi$, and $\sigma_j$, with $j\in \{x,z\}$, are the usual Pauli matrices.
We want to construct a partitioned unitary quantum cellular automata (PUQCA) whose dynamics, in the continuous limit, is the same as the Dirac equation~\eqref{eq:Dirac}.

A PUQCA is a discrete model of quantum computation, which is most easily understood in its graph perspective. We start by discretizing the space coordinate, and thus constructing a 1D linear graph--see figure ~\ref{fig:graphRep}$a$. Each vertex in this graph represents a point in space, and we will assume that the distance between two vertices is $\Delta x$. In this way, the vertices are located at positions $x= k\, \Delta x$, with $k\in \Ir$.

\begin{figure}
	\includegraphics[trim=-10 -20 -60 -50,clip,scale=0.4]{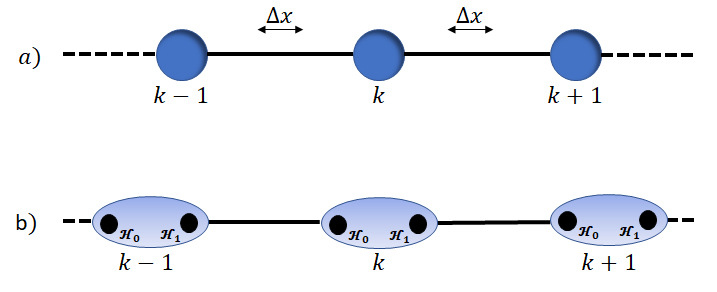}\\
	\caption{
		{\label{fig:graphRep}$a$ Graph representation of the space coordinates. Labeling the vertex positions by $k$ their positions in space can be calculated, i.e., $x=k\Delta x$. $)$ The vertices are related to the PUQCA cells, with two subcells each. }
	}
\end{figure}

In order to describe the dynamics of the particle, an excitation moving in the space, each vertex is split into two subcells and in each subcell, we place a qubit--see figure ~\ref{fig:graphRep}$b$. More concretely, to the vertex at position $x$  we associate a Hilbert space  $(\mathcal{H})_x=(\mathcal{H}_0)_x\otimes(\mathcal{H}_1)_x$ with ${\cal H}_0\simeq{\cal H}_1\simeq \Cx^2$, for the left and right subcells, respectively.  The state $\ket{0}$ will be associated with no particles, while the state $\ket{1}$ will be associated with the presence of a particle/excitation.

In this model, the evolution is also discrete. Defining the time steps size as $\Delta t$,  the values for time are $t= m\, \Delta t$, with $m\in \Nl$. For the cell structure defined, the transition function $\mathcal{E}$, i.e., the operator that updates the PUQCA state $\mathcal{E}\ket{\Psi(t)}= \ket{\Psi(t+\Delta t)}$, is composed by two parts. First, each vertex is individually and simultaneously updated by the application of a unitary $(W_0)_x:(\mathcal{H}_0)_x\otimes(\mathcal{H}_1)_x \rightarrow (\mathcal{H}_0)_x\otimes(\mathcal{H}_1)_x$. Then, to allow for the particle propagation a second unitary is applied, but now acting on the graph edges, i.e., $(W_1)_x^{x+\Delta x}:(\mathcal{H}_1)_x\otimes(\mathcal{H}_0)_{x+\Delta x} \rightarrow (\mathcal{H}_1)_x\otimes(\mathcal{H}_0)_{x+\Delta x}$. The transition function is then
\begin{equation}
\mathcal{E}= \bigotimes_{k\in\Ir} (W_1)_{k\Delta x}^{(k+1)\Delta x} \bigotimes_{k^\prime\in\Ir} (W_0)_{k^\prime \Delta x}.
\end{equation}

Following the construction of a quantum walk that simulates the 1+1 Dirac equation~\cite{PhysRevA.97.042131}, and the translation between quantum walks and the PUQCA~\cite{QCAviaQws}, we choose the unitary transformations as follows. For $W_0$ we allow, a rotation in the single-excitation subspace for each cell:
\begin{align}
\label{eq:W0}
	W_0&= |00\rangle\langle 00|- \ii \sin\theta  |01\rangle\langle 01|   + \cos\theta   |01\rangle\langle 10|\nonumber\\
	&\quad+  \cos\theta |10\rangle\langle 01|	- \ii \sin\theta  |10\rangle\langle 10| +  |11\rangle\langle 11|.
\end{align}
The restriction to the single excitation subspace guarantees that the number of particles is fixed to one. To interact the cells, we employ a simple \textsc{swap} gate between the subcells connected by an edge. That means, we take:
\begin{align}
W_1= |00\rangle\langle 00|+|01\rangle\langle 10|
+|10\rangle\langle 01|+|11\rangle\langle 11|.
\end{align}

Given the above construction, we can now see the evolution on the  one-particle sector of this QCA. To do that we write the system state at time $t$, $\ket{\Psi(t)}$ as:
\begin{small}
\begin{align}
\label{eq:onePar}
&\sum_{k\in \Ir}\Bigg[\psi^0(k\Delta x,t) |\ldots (00)_{(k-1)\Delta x} (10)_{k \Delta x} (00)_{(k+1)\Delta x}\cdots\>\nonumber\\
&\quad +\psi^1(k\Delta x,t) |\cdots (00)_{(k-1)\Delta x} (01)_{k \Delta x} (00)_{(k+1)\Delta x}\cdots\> \Bigg].
\end{align}
\end{small}
Applying the evolution $\epsilon$ to this state the following recurrence relation for the amplitudes can be derived:
{\small
\begin{align}
\label{eq:recurrence_relation}
\psi^{0}\left(x,t+\Delta t\right)&=\cos \theta\psi^{0}\left(x-\Delta x,t\right)-\ii\sin \theta\psi^1\left(x-\Delta x,t\right)\nonumber\\
\psi^{1}\left(x,t+\Delta t\right)&=\cos \theta\psi^1\left(x+\Delta x,t\right) - \ii\sin \theta\psi^{0}(x+\Delta x,t).
\end{align}}
Associating $\psi_0$ and $\psi_1$ with the spinor components of $\psi$, the last step is just to take the continuous limit. Making $\theta = mc^2\Delta t/\hbar$, and taking the limit $\Delta x\rightarrow 0$ and $\Delta t\rightarrow 0$, such that $\Delta x/\Delta t\rightarrow c$, in first order we get  the massive Dirac equation in (1+1)-dimensional spacetime~\eqref{eq:Dirac}. We call such quantum automaton as the Dirac PUQCA.

\section{Coarse-graining a 1D PUQCA}\label{sec:CG map}

Now that we have established the quantum cellular automaton that simulates the 1+1 Dirac equation, we use the automaton spatial structure to define a coarse-graining map. The physical idea is that an observer with low resolution will not be able to discriminate between different automaton cells. 

To incorporate this idea in our approach, we employ completely positive and trace preserving (CPTP) maps, a.k.a. quantum channels, to throw away inaccessible information in a general way. This approach has been developed and used in other  scenarios~\cite{pedrinho,Oleg,mermin1980,isadora2020,CG_QM,cris2019,alicki2009,poulin2005,caslavLG,Wang2013,Jeong2014,silva2020macro}. In general, a coarse-graining channel will be defined as a quantum channel which is a trace preserving and completely positive map that reduces the description of the system: $\Lambda:\mathcal{L}(\mathcal{H}_D)\rightarrow\mathcal{L}(\mathcal{H}_d)$ with $D> d$, where  $\mathcal{L}(\mathcal{H})$ is the set of all linear operators acting on $\mathcal{H}$.

Here, we will apply a coarse-graining channel in each cell of the automaton, reducing their two subcells into a single subcell. Locally the coarse-graining map acts as follows:
\begin{equation}
\label{eq:localCG}
\Lambda_{k\Delta x}:\mathcal{L}\left(\left(\mathbb{C}_{0}^{2}\right)_{k\Delta x} \otimes \left(\mathbb{C}_{1}^{2}\right)_{k\Delta x}\right) \rightarrow \mathcal{L}\left(\left(\mathbb{C}_{ k\hspace*{-0.2cm}\mod 2 }^{2}\right)_{\lfloor{\frac{k}{2}}\rfloor\Delta x}\right).
\end{equation}
The application of the coarse-graining map in the whole automaton is then given by:

\begin{equation}
\label{eq:CG-map}
\Lambda_{CG}=\bigotimes_{k\in \Ir} \Lambda_{k\Delta x}.
\end{equation}
Notice that the application of the coarse-graining map preserves the automaton local structure. It is also important to realize that after ``compressing'' the two subcells at a position $x$ into an effective subcell, this subcell is merged with a neighboring effective subcell to form an effective cell. The distance between two effective cells is still $\Delta x$.  By this procedure, we establish a new level of description of the automaton, with a new 1-d lattice. 

To finalize our procedure to obtain an effective description, we must define the action of $\Lambda_x$.  To model the lack of resolution, we will employ a coarse-graining map that cannot distinguish if an excitation is in the left or right subcell of a point $x$. The action of $\Lambda_x$, which was first suggested in~\cite{CG_QM} in the context of atoms in a optical lattice, is as follows:

\begin{equation}
\begin{tabular}{lc|cl}
$\Lambda_{x}\left(\borb{01}{01}\right)=\borb{1}{1}$&&&$\Lambda_{x}\left(\borb{01}{00}\right)=\borb{1}{0}/\sqrt{3}$\\
$\Lambda_{x}\left(\borb{10}{10}\right)=\borb{1}{1}$&&&$\Lambda_{x}\left(\borb{10}{00}\right)=\borb{1}{0}/\sqrt{3}$\\
$\Lambda_{x}\left(\borb{00}{00}\right)=\borb{0}{0}$&&&$\Lambda_{x}\left(\borb{10}{01}\right)=0$
\end{tabular}
\label{eq:rules}
\end{equation}

The first column contains the diagonal elements and expresses the inability to distinguish whether the excitation is on the left or the right subcell. If no excitation is present, then effectively no excitation is observed. It is clearly observed that this map is not equivalent to a partial trace. 
The second column contains the off-diagonal terms. As the the base vectors in the excited subspace cannot be discriminated, then there can be no coherence among them. The only possible coherence terms are in between the no excitation subspace and the single-excitation subspace. The factor $1/\sqrt{3}$ guarantees the complete positivity of $\Lambda_x$.
As quantum channels preserve the Hermiticity, the action on the Hermitian conjugated terms can be apprehended from~\eqref{eq:rules}. Further discussion on this map and its action on two-particle subspace can be found in~\cite{CG_QM,pedrinho}. 

Given the map $\Lambda_x$, we are now ready to find the coarse-grained 1D PUQCA. As the coarse-graining channel maps the 1-d lattice to another similar 1-d lattice, we can apply this map $L$ times on the underlying automation to obtain  its  $L$th level of description. In the one-particle sector, we can  write the automaton state at level $L$ of description as 
\begin{equation}
\label{eq:mainOp}
    \rho_L = \sum_{k, k^\prime \in \Ir}\sum_{a,a^\prime=0}^1 (\rho_L)^{k \Delta x,a}_{k^\prime \Delta x, a'}\ket{(k,a)}\bra{(k^\prime, a^\prime)},
\end{equation}
where $\ket{(k,a)}$ represents the state with an excitation at position $k \Delta x$ and subcell $a$. The state at level $L+1$, i.e., after we apply the coarse-graining map $\Lambda_{CG}$ in the whole automaton, is given by:

\begin{align}
&\rho_{L+1}=\sum_{k,k^\prime\in \Ir}\sum_{a,a=0}^{1}\left(\delta_{k,k'}\delta_{a,a'}+\frac{\left(1-\delta_{k,k'}\right)}{3}\right)  \nonumber\\ 
&\quad\times(\rho_L)^{k \Delta x,a}_{k^\prime \Delta x, a'}\left|\left(\left\lfloor{\frac{k}{2}}\right\rfloor, k\;\text{mod}\; 2\right)\right\rangle \left\langle\left(\left\lfloor{\frac{k'}{2}}\right\rfloor, k'\;\text{mod}\; 2\right)\right |
\label{eq:cg_L+1}.
\end{align}
The derivation of this expression is done in Appendix~\ref{app:effective}.

From the above expression, we can derive, by recursion, the density matrix coefficients for level $L$ in terms of the density matrix coefficients of level 0:
\begin{equation}
\label{eq:cg0toLpop}
\left(\rho_{L}\right)_{x,b}^{x,b}=\sum_{s=0}^{2^{L-1}-1}\sum_{a=0}^{1}\left(\rho_{0}\right)_{2^L x+2^{L-1} b \Delta x+s\Delta x,a}^{2^L x+2^{L-1} b \Delta x+s\Delta x,a},
\end{equation}
and 

\begin{equation}
\label{eq:cg0toLcoe}
\left(\rho_{L}\right)_{x',b'}^{x,b}=\frac{1}{3^L}\sum_{s,s'=0}^{2^{L-1}-1}\sum_{a,a'=0}^{1}\left(\rho_{0}\right)_{2^L x'+2^{L-1}b'\Delta x+s'\Delta x,a'}^{2^L x+2^{L-1}b\Delta x+s\Delta x,a},
\end{equation}
for $x \neq x'$ and $b,\; b'\in\{0,1\}$. From the recursive relation above for the off-diagonal elements of $\rho_L$, we show in Appendix~\ref{APP:decoh} that $\left|\left(\rho_{L}\right)_{x',b'}^{x,b}\right| \leq (2\sqrt{2}/3)^L$, i.e.  the coherence elements converge exponentially fast to zero with the coarse-graining level. To illustrate this exponential decay, consider the case where at $L=0$ we start the automaton in the state $\ket{\Psi(0)} = (\ket{(0,0)}+\ket{(0,1)})/\sqrt{2}$ and we evolve it with the Dirac PUQCA, setting $\theta=\pi/4$ in equation (\ref{eq:W0}), for 200 time steps. Figure~\ref{fig:CGRd} shows how the sum of all absolute values of the off-diagonal terms decays with level of coarse-graining.

\begin{figure}
	\includegraphics[trim=120 0 130 0,clip,scale=0.125]{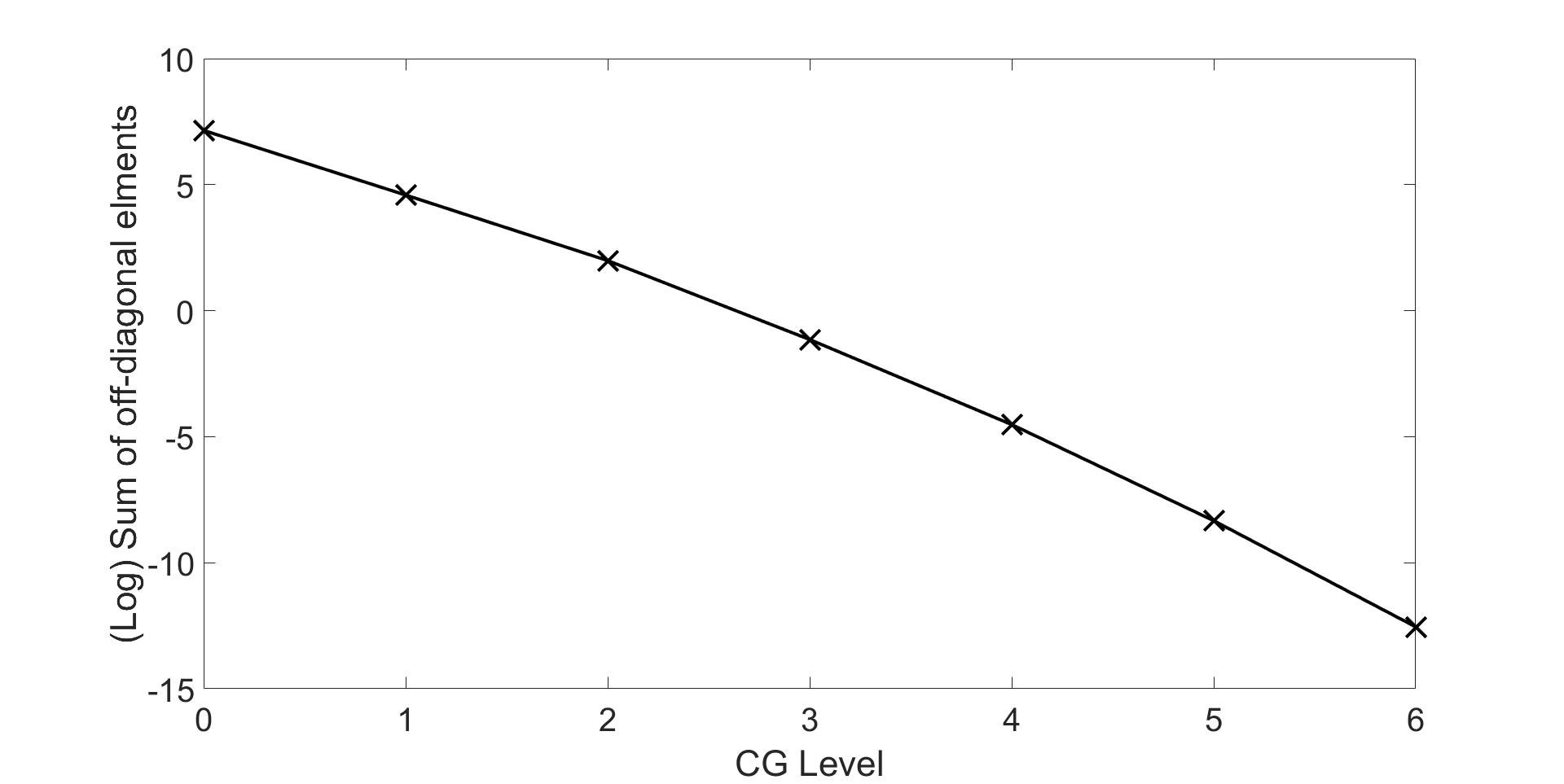}\\
	\caption{
		\label{fig:CGRd}  This plot shows the semilog sum of the off-diagonal elements in different levels of CG for a fixed time. Starting with a QCA state at $t=0$ in $L=0$, the sum of the absolute values of the coherence elements is taken from $L=0$ to $L=6$. An exponential decay behavior can be noticed.}
\end{figure}

\begin{figure}
	\includegraphics[trim=-5 10 10 0,clip,scale=0.5]{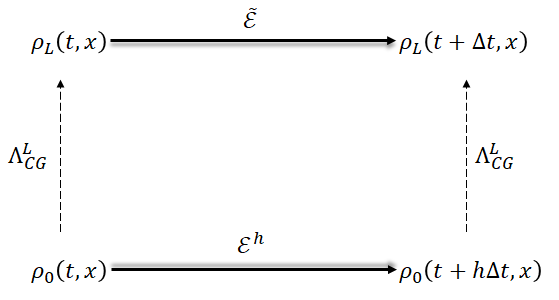}\\
	\caption{\label{fig:CG} The coarse-graining procedure for an arbitrary level of $L$ and time step $h\Delta t$. While at the level $L=0$ the time evolution operator, $\mathcal{E}$, is given by a unitary map, at the level $L$ the time evolution is given by a quantum channel $\tilde{\mathcal{E}}$ derived from $\mathcal{E}$.}  
\end{figure}

\section{Coarse-grained Dirac PUQCA}
\label{sec:CGQCA}

In the previous section, we computed the coarse-grained density operator for an arbitrary level $L$ and we wrote the recursive equations over different $L$. In this section, we will compute the coarse-grained QCA. Before the coarse-graining map, the QCA is governed by equation (\ref{eq:onePar}) and it is well known that such evolution models a free evolution of Dirac particles, as we recalled in the first section. Now, what can we say about the dynamics at a generic level $L$? We certainly know that the QCA will lose coherence at each application of the coarse-graining map. For sake of simplicity, we start restricting ourselves to $\theta = 0$, the \textit{massless case}. This means, from equation (\ref{eq:recurrence_relation}), that $\psi^0$ and $\psi^1$ are fully decoupled: left(right)-moving modes at level $L=0$ maps to left(right)-moving modes at level $L\geq 1$. Thus, let us focus on the population terms. Before we coarse-grain the density operator at time $t$, we add a new  information about the initial condition. Here we assume that we start with a superposed state in the center of the lattice of size $2^L+1$, that is for $l=0,1,\ldots, L-1$
\begin{align}
\left(\rho_0(0)\right)^{2^{L-1-l},0}_{2^{L-1-l},0}&=\alpha^{(0)}_l,\\
\left(\rho_0(0)\right)^{2^{L-1+l},1}_{2^{L-1+l},1}&=\alpha^{(1)}_l,\nonumber
\end{align}
where $\alpha^{(1)}_l, \alpha^{(0)}_l \in \mathbb{R}$ and $\sum_{l=0}^{L-1}(\alpha^{(0)}_l+\alpha^{(1)}_l)=1$. Now, for a given time $t\ll2^L$ from the massless dynamics, equation (\ref{eq:recurrence_relation}), since any superposition is created  anymore when $\theta=0$, all amplitudes that started on the left(right) of the center remain in the subcells 0(1) that is

\begin{eqnarray}
\label{eq:level_L_at_T}
\left(\rho_L(t)\right)^{x,0}_{x,0}&=&\sum_{s=0}^{2^{L-1}-1}\left(\rho_0(t)\right)^{2^Lx+s\Delta x,0}_{2^Lx+s\Delta x,0}\\
\left(\rho_L(t)\right)^{x,1}_{x,1}&=&\sum_{s=2^{L-1}}^{2^{L}-1}\left(\rho_0(t)\right)^{2^Lx+s\Delta x,1}_{2^Lx+s\Delta x,1}.
\end{eqnarray}
Now, let us coarse-graining the density operator $\rho$ at time $t+\Delta t$ and then use equation ($\ref{eq:recurrence_relation}$) for $\theta=0$:
\begin{eqnarray}
\label{eq:pho0}
\rho^{(0)}_L(t+\Delta t,x)&=&\sum_{s=0}^{2^{L-1}-1}\left(\rho_0(t+\Delta t)\right)^{2^Lx+s\Delta x,0}_{2^Lx+s\Delta x,0}\\
&=&\sum_{s=0}^{2^{L-1}-1}\left(\rho_0(t)\right)^{2^Lx-\Delta x+s\Delta x,0}_{2^Lx-\Delta x+s\Delta x,0},\nonumber\\
\rho^{(1)}_L(t+\Delta t,x)&=&\sum_{s=0}^{2^{L-1}-1}\left(\rho_0(t+\Delta t)\right)^{2^Lx+s\Delta x,1}_{2^Lx+s\Delta x,1}\label{eq:pho1}\\
&=&\sum_{s=2^{L-1}}^{2^{L}-1}\left(\rho_0(t)\right)^{2^Lx+\Delta x+s\Delta x,1}_{2^Lx+\Delta x+s\Delta x,1},\nonumber
\end{eqnarray}
where we adopted the notation $\rho^{(i)}_L(t+\Delta t,x)=\left(\rho_L(t)\right)^{x,i}_{x,i}$ for $i=0,1$ to describe the the density matrix elements in the level $L$. The challenge here is to see whether in the right-hand side of equations (\ref{eq:pho0}) and (\ref{eq:pho1}), there is a density matrix in the level $L$ at time $t$ which is derived from the a density matrix of level 0 at time $t$. It turns out that the diagonal elements of the density matrix of level 0 after we coarse-grained $L$ times are $\rho^{(i)}_L(t,x-(-1^i)\Delta x/2^L)$, and therefore,
\begin{eqnarray}
\label{eq:recurrence_updated}
\rho^{(0)}_L(t+\Delta t,x)&=&\rho^{(0)}_L(t,x-\Delta x/2^L)\\
\rho^{(1)}_L(t+\Delta t,x)&=&\rho^{(1)}_L(t,x+\Delta x/2^L)\nonumber.
\end{eqnarray}

Now, Taylor expands the above equation around $x$ and $t$ up to a first order, and the following continuous differential equations can be established 
\begin{eqnarray}
\label{eq:recurrence_updated_Dif}
\partial_t\rho^{(0)}_L(t,x)&=& -\frac{c}{2^L}\partial_{x} \rho^{(0)}_L(t,x) \\
\partial_t\rho^{(1)}_L(t,x)&=& \frac{c}{2^L}\partial_{x} \rho^{(1)}_L(t,x) \nonumber.
\end{eqnarray}
where $c=\Delta x/\Delta t$. From the results above, it is clear that the factor $1/2^L$ starts to create a stationary state as we increase the level $L$, see figures (\ref{fig:CGspee}).  In these plots, we can see the contrast between the light cones after the CG map generated from a massless particle which started in a superposition state in $x=0$ at $t=0$. In other words, by doing only CG in space, trivial dynamics start to appear fast. Inspired by the results given in \cite{costaPCA} more interesting cases can be derived when the coarse-graining in time is also done. In this procedure, the transition function, $\mathcal{E}$, can be applied  $h$ times before each CG application. The number of $\mathcal{E}$ applications satisfies the following upper bound $h \leq s$, where $s$ is the amount of grouped subcells where the local maps $\Lambda_x$ are applied. In our case, from equation (\ref{eq:localCG}), $s=2$. 

\begin{figure}
	\includegraphics[trim=350 0 80 0,clip,scale=0.18]{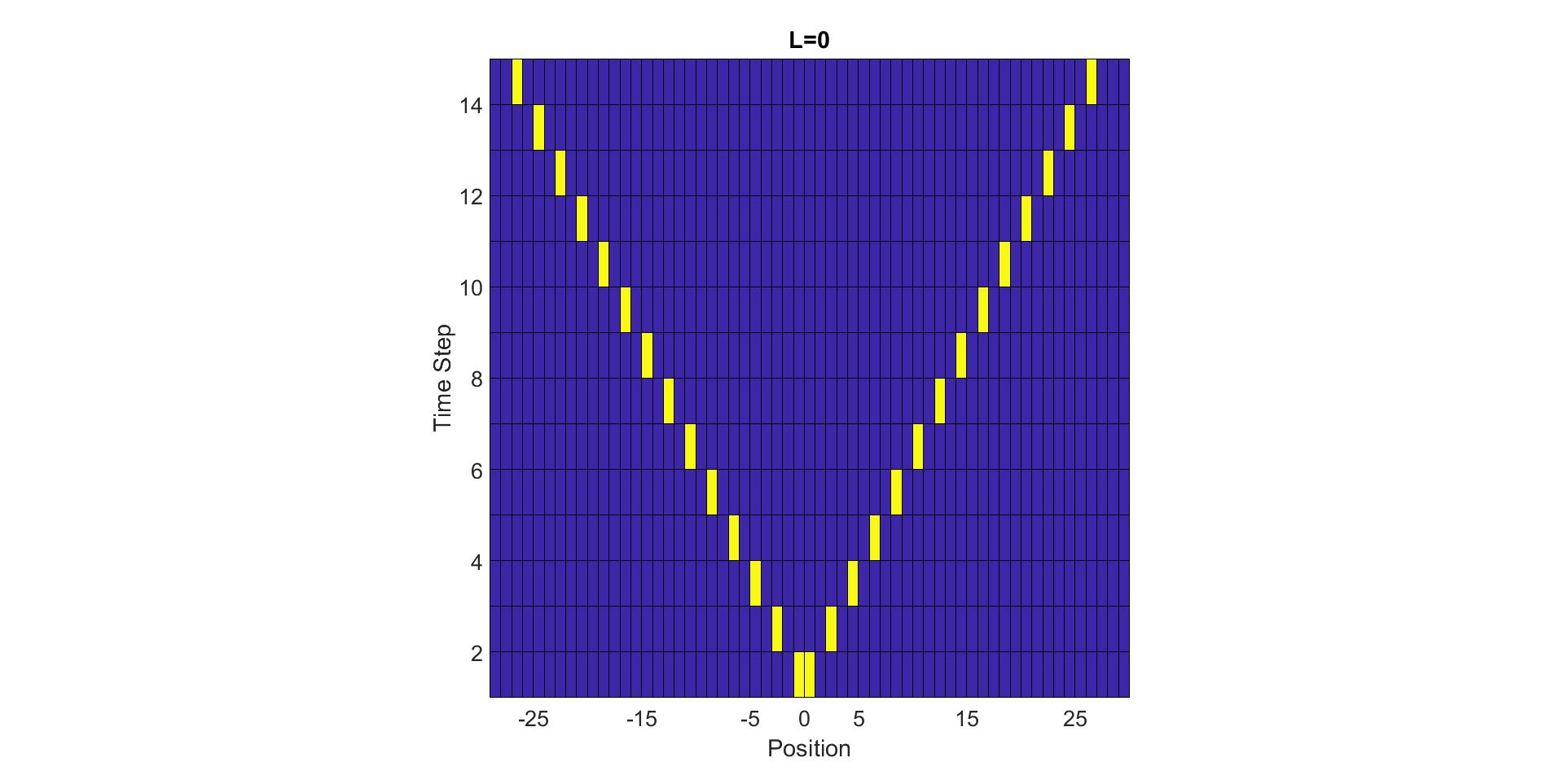}\\ \quad(a)
	\includegraphics[trim=350 0 80 0,clip,scale=0.18]{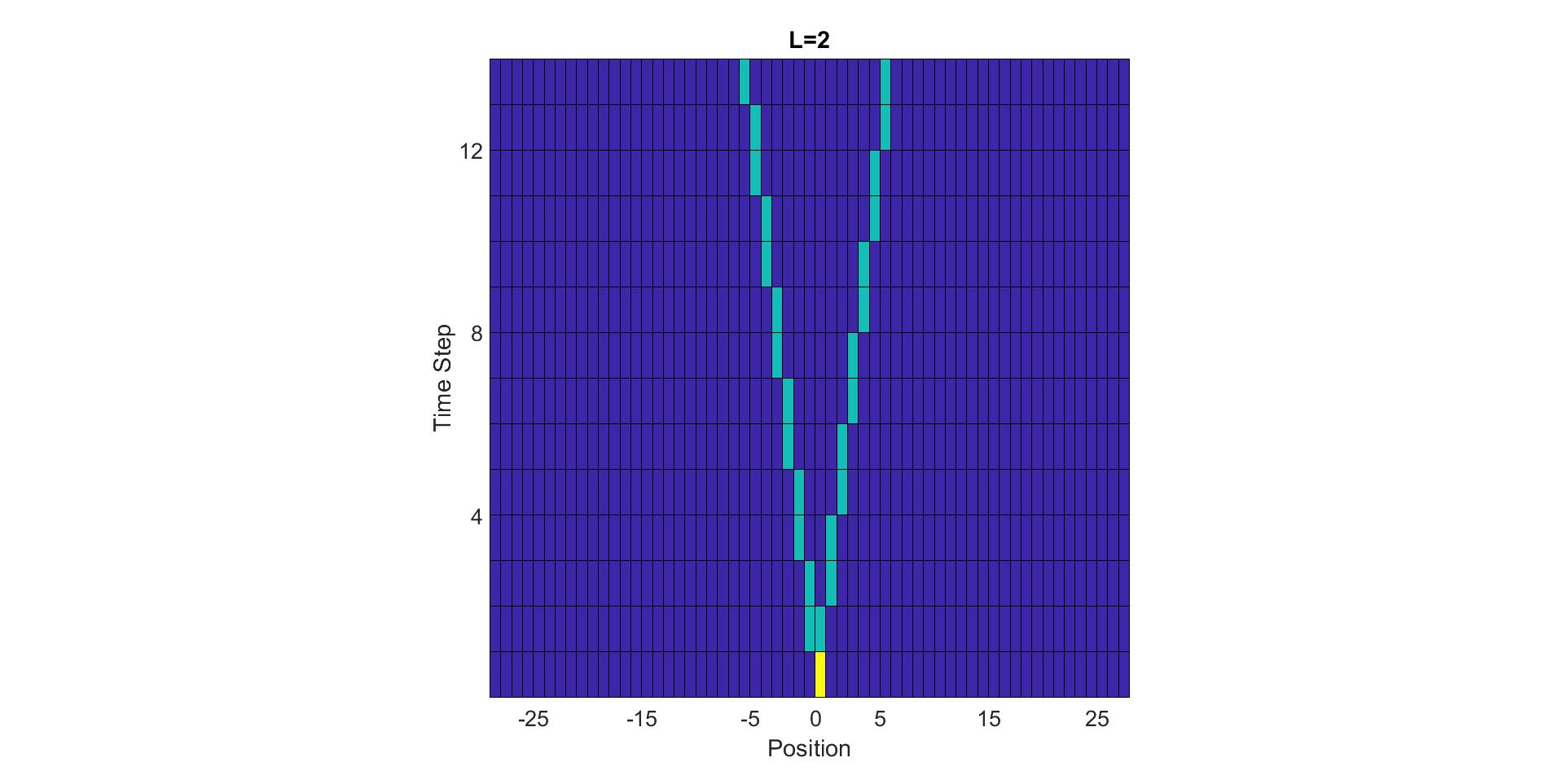}\\ \quad(b)
	\caption{
		{\label{fig:CGspee}  Light cone propagation before (a) and after (b) two space coarse-grained transformation. The particle propagation becomes slower at each CG transformation, where the speed of the particle is rescaled by the factor $1/2^L$. }
	}
\end{figure}

In figure (\ref{fig:CG}) we elucidated the procedure with an arbitrary value for $h$. Notice that in the new level the update state is achieved by only one application of $\tilde{\mathcal{E}}$, i.e.,  the quantum channel emerged from the CG process. However, in the level $L$ we rescaled the time from $t+h\Delta t$ to $t+\Delta t$ when $h\neq 1$.  As explained in~\cite{costaPCA}, we do not accept $h>s$ since we want to preserve the local structure in the emergent dynamics which comes from the light cone existent in the QCA evolution.

Let us go back to our previous problem, given by equation (\ref{eq:recurrence_updated_Dif}), where the localization starts to appear. Combining the temporal with the spatial CG, this no physical situation generated by the particle localization can be avoided. In fact, we can keep with the same light cone structure if before each CG operation we compensate it by applying the transition function twice. Thus, since we are moving from the level 0 straight to the level $L$, which is equivalent of we group $2^L$ cells, the velocity can be kept constant if the compensation by $h=2^L$ is done before $\Lambda_{CG}^{L}$. In other words, we have $t+2^L\Delta t$ in the level 0 equivalent to  $t+\Delta t$ in the level $L$, 
\begin{eqnarray*}
\left(\rho_L(t+\Delta t)\right)^{x,0}_{x,0}&=&\sum_{s=0}^{2^{L-1}-1}\left(\rho_0(t+2^L\Delta t)\right)^{2^Lx+s\Delta x,0}_{2^Lx+s\Delta x,0}\\
\left(\rho_L(t+\Delta t)\right)^{x,1}_{x,1}&=&\sum_{s=2^{L-1}}^{2^{L}-1}\left(\rho_0(t+2^L\Delta t)\right)^{2^Lx+s\Delta x,1}_{2^Lx+s\Delta x,1}.
\end{eqnarray*}
Therefore, $x \pm \Delta x/2^L$, goes to $x \pm \Delta x$, when we have  $t+2^L\Delta t$ in the level 0 instead $t+\Delta t$, and consequently the factor $1/2^L$ in equation (\ref{eq:recurrence_updated_Dif}) is removed, that is
\begin{eqnarray}
\label{eq:recurrence_updated_Dif2}
\partial_t\rho^{(0)}_L(t,x)&=& -c\partial_{x} \rho^{(0)}_L(t,x)\\
\partial_t\rho^{(1)}_L(t,x)&=& c\partial_{x} \rho^{(1)}_L(t,x)\nonumber.
\end{eqnarray}

Before we move to the case where $m\neq 0$, here we can already point out that the emergent dynamics derived in the classical limit can be easily described by the partitioned cellular automata~\cite{costaPCA,Wol00,Margolus} with two subcells, which is equivalent to two bits at each cell in the classical case. The transition function is composed of two permutation operators. Both permutations are the SWAP operator, analogous to the PUQCA used to describe the massless case, that is, the first one acting locally inside all vertices, moving the excitation from the subcell 0 (1) to the he subcell 1 (0), and the second interacting the vertices, sending the excitation from the vertex $x$, in the subcell 1 (0) to $x+1$ $(x-1)$ in the subcell 0 (1).

The starting point for the massive case is the recurrence relation for the probabilities. From equation (\ref{eq:recurrence_relation}) when $\sin\theta \approx \theta$, in terms of the density matrix elements, we have
\begin{align}
\label{eq:recurrence_relatioDens}
\rho(t+\Delta t)_{x,0}^{x,0} &\approx \rho(t)_{x-\Delta x,0}^{x-\Delta x,0}\\
&\quad+i\frac{mc^2\Delta t}{\hbar}\left( \rho(t)_{x-\Delta x,1}^{x-\Delta x,0}-\rho(t)_{x-\Delta x,0}^{x-\Delta x,1}\right)\nonumber\\
\rho(t+\Delta t)_{x,1}^{x,1} &\approx \rho(t)_{x+\Delta x,1}^{x+\Delta x,1}\nonumber\\
&\quad+i\frac{mc^2\Delta t}{\hbar}\left( \rho(t)_{x+\Delta x,0}^{x+\Delta x,1}-\rho(t)_{x+\Delta x,1}^{x+\Delta x,0}\right),\nonumber
\end{align}
Notice that in the pair of equations above the second-order terms in time, i.e., $\Delta t^2$, are not included. Now, we will not assume anything about the initial condition, and thus, in the level $L$ we have
\begin{align*}
\left(\rho_L(t)\right)^{x,0}_{x,0}&=\sum_{s=0}^{2^{L-1}-1}\left[\left(\rho_0(t)\right)^{2^Lx+s\Delta x,0}_{2^Lx+s\Delta x,0}+\left(\rho_0(t)\right)^{2^Lx+s\Delta x,1}_{2^Lx+s\Delta x,1}\right]\\
&=\sum_{s=0}^{2^{L-1}-1}\left(\rho_0(t)\right)^{2^Lx+s\Delta x}_{2^Lx+s\Delta x},\\
\left(\rho_L(t)\right)^{x,1}_{x,1}&=\sum_{s=2^{L-1}}^{2^{L}-1}\left[\left(\rho_0(t)\right)^{2^Lx+s\Delta x,0}_{2^Lx+s\Delta x,0}+\left(\rho_0(t)\right)^{2^Lx+s\Delta x,1}_{2^Lx+s\Delta x,1}\right],\\
&=\sum_{s=2^{L-1}}^{2^{L}-1}\left(\rho_0(t)\right)^{2^Lx+s\Delta x}_{2^Lx+s\Delta x},\\
\end{align*}
wherein the dependence of the internal state of the particle was eliminated in the right part of the equality above. This elimination can be done since the probability of finding a particle in the cell $x$ is given by $\rho(t)^{x,0}_{x,0}+\rho(t)^{x,1}_{x,1}$. If we also add both equations above, we get
\begin{equation}
\label{eq:Mass}
    \left(\rho_L(t)\right)^{x}_{x}=\sum_{s=0}^{2^{L}-1}\left(\rho_0(t)\right)^{2^Lx+s\Delta x}_{2^Lx+s\Delta x}.
\end{equation}
Likewise the massless case, our next step is to look up our density matrix at time $t+\Delta t$ and then making use of equations (\ref{eq:recurrence_relatioDens}). Thus, the update state at level $L$ writing in terms of the level 0 is 
\begin{align}
\label{eq:p0}
&\left(\rho_L(t+\Delta t)\right)^{x,0}_{x,0}\nonumber\\ &\approx \sum_{s=0}^{2^{L-1}-1}\left[\left(\rho_0(t)\right)^{2^Lx-\Delta x+s\Delta x,0}_{2^Lx-\Delta x+s\Delta x,0}+\left(\rho_0(t)\right)^{2^Lx+\Delta x+s\Delta x,1}_{2^Lx+\Delta x+s\Delta x,1}\right.\nonumber\\
&\quad+i\frac{m c^2\Delta t}{\hbar}\left(\left(\rho_0(t)\right)^{2^Lx-\Delta x+s\Delta x,0}_{2^Lx-\Delta x+s\Delta x,1}-\left(\rho_0(t)\right)^{2^Lx-\Delta x+s\Delta x,1}_{2^Lx-\Delta x+s\Delta x,0}\right.\nonumber\\
&\quad+\left.\left.\left(\rho_0(t)\right)^{2^Lx+\Delta x+s\Delta x,1}_{2^Lx+\Delta x+s\Delta x,0}-\left(\rho_0(t)\right)^{2^Lx+\Delta x+s\Delta x,0}_{2^Lx+\Delta x+s\Delta x,1}\right)\right],
\end{align}
\begin{align}
&\left(\rho_L(t+\Delta t)\right)^{x,1}_{x,1}\nonumber\\ &\approx \sum_{s=2^{L-1}}^{2^{L}-1}\left[\left(\rho_0(t)\right)^{2^Lx-\Delta x+s\Delta x,0}_{2^Lx-\Delta x+s\Delta x,0}+\left(\rho_0(t)\right)^{2^Lx+\Delta x+s\Delta x,1}_{2^Lx+\Delta x+s\Delta x,1}\right. \nonumber\label{eq:p1}\\
&\quad+i\frac{m c^2\Delta t}{\hbar}\left(\left(\rho_0(t)\right)^{2^Lx-\Delta x+s\Delta x,0}_{2^Lx-\Delta x+s\Delta x,1}-\left(\rho_0(t)\right)^{2^Lx-\Delta x+s\Delta x,1}_{2^Lx-\Delta x+s\Delta x,0}\right.\nonumber\\
&\quad+\left.\left.\left(\rho_0(t)\right)^{2^Lx+\Delta x+s\Delta x,1}_{2^Lx+\Delta x+s\Delta x,0}-\left(\rho_0(t)\right)^{2^Lx+\Delta x+s\Delta x,0}_{2^Lx+\Delta x+s\Delta x,1}\right)\right].
\end{align}
In this stage, we can already do a Taylor expansion in both sides. On the left side, we expand until the first order in time. However, on the right side there is a subtlety that should be noticed. First, the expansion is done around the point $2^Lx+p$, where $p=\pm s\Delta x$. Second, while for the terms without the mass the expansion is done until the first order, that is
\begin{align}
&\left(\rho_0(t)\right)^{2^Lx-\Delta x+s\Delta x,0}_{2^Lx-\Delta x+s\Delta x,0}+\left(\rho_0(t)\right)^{2^Lx+\Delta x+s\Delta x,1}_{2^Lx+\Delta x+s\Delta x,1}\\ 
&\quad\approx \left(\rho_0(t)\right)^{2^Lx+s\Delta x}_{2^Lx+s\Delta x}-\frac{\Delta x}{2^L}\partial_x\left(\rho_0(t)\right)^{2^Lx+s\Delta x,0}_{2^Lx+s\Delta x,0}\nonumber\\
&\quad+\frac{\Delta x}{2^L}\partial_x\left(\rho_0(t)\right)^{2^Lx+s\Delta x,1}_{2^Lx+s\Delta x,1},\nonumber
\end{align}
the ones where the mass appears we only consider the zero order in the expansion. This difference in the order considered during the Taylor expansion appears because the ones which includes the mass are already first order terms in $\Delta t$. However, it turns out that the result of adding all zero orders terms eliminates the mass in the emergent dynamics, since
\begin{equation}
\left(\rho_0(t)\right)^{2^Lx-\Delta x+s\Delta x,0}_{2^Lx-\Delta x+s\Delta x,1}-\left(\rho_0(t)\right)^{2^Lx+\Delta x+s\Delta x,0}_{2^Lx+\Delta x+s\Delta x,1} \approx O\left(\frac{\Delta x}{2^L}\right).
\end{equation}
Thus, the following dynamics is derived for the massive case
\begin{align}
\label{eq:masscase}
   &\partial_t\left(\rho_L(t)\right)^{x}_{x}\nonumber\\
   &=\sum_{s=0}^{2^{L}-1}\frac{c}{2^L}\partial_x\left(\left(\rho_0(t)\right)^{2^Lx+s\Delta x,1}_{2^Lx+s\Delta x,1}-\left(\rho_0(t)\right)^{2^Lx+s\Delta x,0}_{2^Lx+s\Delta x,0}\right).
\end{align}
Therefore, all the interference terms that would appear as a consequence of the mass are not detectable in the level $L$. Moreover, since in equation (\ref{eq:masscase}) the dynamics of the level $L$ is level 0 dependent, there is not a model in the level $L$ which is derived from the level 0 when the mass of the particle is included. 
 
However, in the level $L$ the internal degree of the particle  is not perceived anymore, and then, since the states 0 and 1 belong to the same cell but located in different points, as shown in figure (\ref{fig:graphRep}b), we can say that the displacement between them is $\Delta x$ as well. Therefore, by setting 
\begin{equation}
\label{eq:equiv}
  \left(\rho_0(t)\right)^{x,1}_{x,1} = \left(\rho_0(t)\right)^{x}_{x},
\end{equation}
the following equality holds
\begin{equation}
\label{eq:equiv2}
  \left(\rho_0(t)\right)^{x,0}_{x,0} = \left(\rho_0(t)\right)^{x-\Delta x}_{x-\Delta x}.
\end{equation}
Using this equivalence into equation (\ref{eq:masscase}), and considering that we have $\Delta x^2/\Delta t=\lambda$, we get
\begin{align}
\label{eq:dif}
   \partial_t\rho_L(t,x) &= \sum_{s=0}^{2^{L}-1} \frac{\Delta x^2}{2^L\Delta t} \partial_x^2\left(\left(\rho_0(t)\right)^{2^Lx+s\Delta x}_{2^Lx+s\Delta x}\right)\\
   &= \frac{\Delta x^2}{2^L\Delta t}\partial_x^2\left(\sum_{s=0}^{2^{L}-1}\left(\rho_0(t)\right)^{2^Lx+s\Delta x}_{2^Lx+s\Delta x}\right)\nonumber\\
   &=\frac{\lambda}{2^L}\partial_x^2\rho_L(t,x).\nonumber
\end{align}
which is the diffusion equation. 

In figure (\ref{fig:Var}) we contrast the variance between the level 0, that simulates the Dirac equation, and the emerged one. In the level 0, we can see the character linear of the variance, as expected from the quantum walks dynamics. However, in the level 3 we can notice a nonlinear behavior, similar to the dispersion relation provided by the random walk. Such result shown in  figure (\ref{fig:Var}) supports the continuous limit derived in equation (\ref{eq:dif}).  

\begin{figure}
	\includegraphics[trim= 160 0 10 40,clip,scale=0.14]{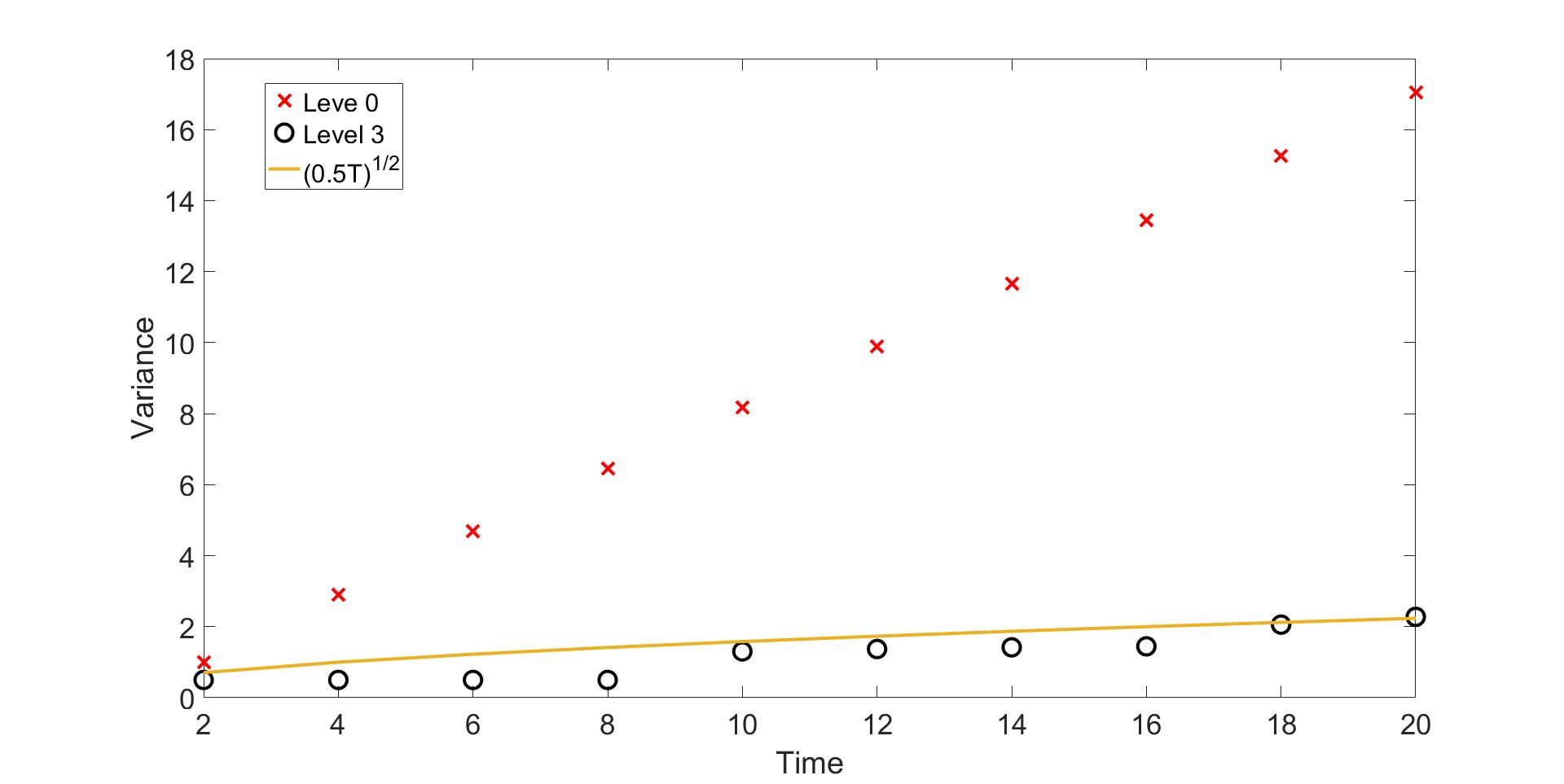}\\
	\caption{\label{fig:Var} Variance against time from the QCA at level 0 and the emergent dynamics derived at level 3. In both cases, the same value $\theta=0.2$ is used.   }  
\end{figure}

\section{Conclusion}\label{sec:conclusion}
In this work, we studied emergent dynamics starting from the Dirac like Hamiltonian, i.e., $H_D=\sigma_x \otimes p_x +m\sigma_z \otimes I$, where the dynamics is driven by $\partial_t\rho=\left[H,\rho\right]$. We used the partitioned unitary quantum cellular automata \cite{QCAviaQws} as the computational model to derive $\exp(i H_D)$. Inspired by \cite{costaPCA}, we provided an analogous coarse-graining scheme to explore emergent dynamics, but now for its quantum version, namely the PUQCA. 

We presented a numerical instance and analytical evidence that the coarse-graining map employed here generates decoherence in our quantum system. As a consequence of multiple CG map applications, a semiclassical system can be achieved. Furthermore, we explicitly show that under a particular initial condition for the massless case, the classical transport given by $\partial_t \rho_L +v\sigma_z \partial_x \rho_L=0$ is the emergent dynamics by computing the continuous limit. In this emergent system, the velocity $v$ can either go exponentially faster to zero, i.e., $v\sim 1/2^L$, or remain the same if the temporal CG is also included. On the other hand, when the mass is included and no assumptions about the initial conditions are done, the diffusion equation is recovered.   It is well known that the diffusion equation is achieved by computing the continuous limit of the random walk. Therefore, the CG procedure developed here provided the right tools to see how the random walk can be derived from the quantum walk. 

Differently from \cite{costaPCA} that shows a general procedure to study emergent dynamics, here just one map and one type unitary operators are used. However, our results are pointing out that this procedure can be generalized for different \textit{complete positive trace-preserving} maps as well as for different dynamics, including the cases where we have interacting quantum walks~\cite{twoWalk1,twoWalk2,twoWalk3}, and we intend to address these possibilities in future publications. Another question that immediately raises from our procedure is whether we can derive emergent local quantum channels $\tilde{\mathcal{E}}$ arising from QCA, following the procedure shown in~\cite{CG_QM,duarte2020investigating}. If such quantum channels can be achieved, we will provide a powerful tool to generate transition functions such that noisy quantum cellular automata, as described in~\cite{NoiseQCA}, naturally appear during the decoherence process. 

Therefore, this formalism might be an excellent candidate to study open quantum systems and also to shed some light on more fundamental questions in the quantum theory, e.g., a tool to explore quantum-to-classical transition via quantum cellular automata. Finally, another possible direction to go is contrasting the emergent dynamics derived here against alternative formulation to study different coarse levels like the use of the wavelets basis~\cite{wavelet}.

\section*{Acknowledgements}

I am really grateful to Fernando de Melo who introduced the subject of coarse-graining and proposed the topic of emergent dynamics via quantum cellular automata during my PhD I also would like to thank  Giuseppe Di Molfetta for helpful discussion during this project, Yuval R. Sanders for careful
reading of the manuscript, and Gavin Brennen for helpful comments. This work was partially supported by ARC Discovery Projects.
 \\

\section{Appendices}

\subsection{Effective state at level $L+1$ from the state at level $L$}
\label{app:effective}

In the following, we apply the coarse-graining local map to the operator (\ref{eq:mainOp}), getting a general expression for the coarse-grained density operator in level $L+1$. 
To get this general form, note that we can distinguish three different cases: (i) $x=x'$ and $a=a'$, (ii) $a\neq a'$, for all for $x = x$, and (iii) $x \neq x$  . In case (i), the state at level $L+1$ reads: 
\begin{align*}
    &\Lambda_{CG}\left(\borb{k,a}{k,a}\right)\\
    &=\left|\cdots\left(0\right)_{k-1}\left(1\right)_{k}\left(0\right)_{k+1}\cdots\right\rangle \left\langle \cdots\left(0\right)_{k-1}\left(1\right)_{k}\left(0\right)_{k+1}\cdots\right|.
\end{align*}
Note that the three conditions for $x$ shown above imply the same for $k$.
The number of subcells for each position is reduced to one, and we have to space group subcells at different spatial positions, to build a new cell hosting two qubits. Consequently, the choice of the space grouping of the subcells, at a higher level, might be determinant and may depend on the spatial position at level $L$. In fact, applying the map at even position $x$ leads to:
\[
\left|\cdots\left(1\right)_{\frac{k}{2}}\cdots\right\rangle \left\langle \cdots\left(1\right)_{\frac{k}{2}}\cdots\right|=\left|\frac{k}{2},0\right\rangle \left\langle \frac{k}{2},0\right|.
\]
where the one particle excitation takes now this place at $a=0$, the left subcell. While if the local map is applied at odd position $x$, the single-particle excitation will be placed in the right subcell $a=1$:
\[
\left|\cdots\left(1\right)_{\frac{k-1}{2}}\cdots\right\rangle \left\langle \cdots\left(1\right)_{\frac{k-1}{2}}\cdots\right|=\left|\frac{k-1}{2},1\right\rangle \left\langle \frac{k-1}{2},1\right|.
\]
Notice once again that, without lack of generality, we space-group from the left to the right. 

Both the above cases can be recast as follows:
\begin{equation}
\Lambda_{CG}\left(\borb{k,a}{k,a}\right)= \borb{\frac{k}{2},k\;\text{mod}\; 2}{\frac{k}{2},k\;\text{mod}\; 2}
\label{eq:ra=r'a'}
\end{equation}

Case (ii), i.e., $a\neq a'$ and $x=x'$ leads directly to 
\begin{equation}
\Lambda_{CG}\left(\left|k,a\right\rangle \left\langle k,a'\right|\right)=0\label{eq:r=r'},
\end{equation}
because $\Lambda_{x}\left(\left|01\right\rangle \left\langle 10\right|\right)=0$.

Finally, for $x \neq x'$, we follow the same procedure, and we apply the local map for each spatial position $x$:
\begin{align*}
    &\Lambda_{CG}\left(\left|k,a\right\rangle \left\langle k',a'\right|\right)=\\
    &\left|\ldots\left(1\right)_{k},\ldots\left(0\right)_{x_L'},\ldots\right\rangle \left\langle \ldots,\left(0\right)_{k},\ldots\left(1\right)_{k'},\ldots\right|,
\end{align*}
and then, from equation (\ref{eq:rules}), we see that the CPTP map covering the coherence terms leads us to
\begin{align}
&\Lambda_{CG}\left(\left|k,a\right\rangle \left\langle k',a'\right|\right)\nonumber\\
&=\frac{1}{3}\left|\left\lfloor\frac{k}{2}\right\rfloor,k\;\text{mod}\; 2\right\rangle \left\langle \left\lfloor\frac{k'}{2}\right\rfloor,k'\;\text{mod}\; 2\right|.
\label{eq:rneqr'}
\end{align}
Finally, recast together equations (\ref{eq:ra=r'a'}), (\ref{eq:r=r'}) and (\ref{eq:rneqr'}), we get the final expression shown in equation (\ref{eq:cg_L+1}), which is the coarse grained density operator at level $L+1$.

\subsection{Decoherence by coarse-graining}\label{APP:decoh}

\begin{lemma}
\label{lemma}
Let $p\in [0,1]^D$, with $D\ge 2$, be a probability distribution, i.e., $\sum_{i=1}^D p_i = 1$. Then the inequality
\begin{equation}
    \sum_{i=1}^D \sqrt{p_i}\le \sqrt{D}
\end{equation}
holds, and it is saturated when $p_i=1/D$ for all $i\in [D]$.
\end{lemma}
\begin{proof}
The proof goes by the way of induction. For the base case, $D=2$, we define the function $f_2: [0,1] \rightarrow \Rl$ as 
\[
f_2(p_1) = \sqrt{p_1} +\sqrt{1-p_1};
\]
where we already employed the normalization property. Setting $d f_2/d p_1=0$ is simple to obtain that the maximum is obtained for $p_1=1/2$. Therefore, $\sqrt{p_1} +\sqrt{1-p_1} \le \sqrt{2}$.

For the inductive step, we assume that $\sum_{i=0}^{D-1} \sqrt{p_i}\le \sqrt{D-1}$. For $D$ we have:
\begin{align}
    \sum_{i=0}^{D} \sqrt{p_i} & = \sum_{i=0}^{D-1} \sqrt{p_i} + \sqrt{p_{D}}\nonumber\\
    & = \sqrt{\sum_{j=0}^{D-1} p_j} \left(\sum_{i=0}^{D-1} \sqrt{\frac{p_i}{\sum_{j=0}^{D-1} p_j}}\right) + \sqrt{p_{D}}.
\end{align}
Noting that  defining $p'_i = \sqrt{p_i/\sum_{j=0}^{D-1} p_j}$, then $p'$ is a probability distribution in $[0,1]^{D-1}$, for which the induction hypothesis can be applied, and that $\sum_{j=0}^{D-1} p_j = 1- p_{D}$ we have:
\begin{equation}
    \sum_{i=0}^{D} \sqrt{p_i}  \le \sqrt{1- p_{D}}\sqrt{D-1} + \sqrt{p_{D}}.
\end{equation}
Finally, the function $f_D(p_D)=\sqrt{1- p_{D}}\sqrt{D-1} + \sqrt{p_{D}}$ has its maximum value, $ \sqrt{D}$, at $p_D=1/D$, what finishes the proof.
\end{proof}
Before we proof that the off-diagonal elements of the density matrix go exponentially faster to zero, the simplest case of the upper bound value from the level 0 to the level 1 is explicitly computed. 

From equation (\ref{eq:cg0toLcoe}), for $b=b'=0$, we see that when we move from the level $0$ to $1$ we have,
\begin{align}
\left(\rho_{1}\right)^{x,0}_{x,0}=&\frac{1}{3}\left[\left(\rho_{0}\right)^{2x,0}_{2x',0}+\left(\rho_{0}\right)^{2x,1}_{2x',1}\right.\\
&\quad+\left.\left(\rho_{0}\right)^{2x,1}_{2x',0}+\left(\rho_{0}\right)^{2x,0}_{2x',1}\right]\nonumber.
\end{align}
Now, from the triangle inequality, i.e., $\left|a+b\right|\leq \left|a\right|+\left|b\right|$, we can write
\begin{align}
\label{eq:ine2}
\left|\left(\rho_{1}\right)^{x,0}_{x',0}\right|\leq&\frac{1}{3}\left[\left|\left(\rho_{0}\right)^{2x,0}_{2x',0}\right|+\left|\left(\rho_{0}\right)^{2x,1}_{2x',1}\right|\right.\\
&\quad+\left.\left|\left(\rho_{0}\right)^{2x,1}_{2x',0}\right|+\left|\left(\rho_{0}\right)^{2x,0}_{2x',1}\right|\right]\nonumber.
\end{align}

Using the following inequality for the density matrix
\begin{equation}
\label{eq:ineq_pop_coh}
\left|\rho^{r,a}_{r',a'}\right|^2\leq \rho_{r,a}^{r,a}\rho_{r',a'}^{r',a'},
\end{equation}

equation (\ref{eq:ine2}) can be written as
{\footnotesize\begin{align*}
\left|\left(\rho_{1}\right)^{x,0}_{x',0}\right|\leq&\frac{1}{3}\left[\sqrt{\left(\rho_{0}\right)^{2x,0}_{2x,0}\left(\rho_{0}\right)^{2x',0}_{2x',0}}+\sqrt{\left(\rho_{0}\right)^{2x,1}_{2x,1}\left(\rho_{0}\right)^{2x',1}_{2x',1}}\right.\\
&\quad+\left.\sqrt{\left(\rho_{0}\right)^{2x,1}_{2x,1}\left(\rho_{0}\right)^{2x',0}_{2x',0}}+\sqrt{\left(\rho_{0}\right)^{2x,0}_{2x,0}\left(\rho_{0}\right)^{2x',1}_{2x',1}}\right]\nonumber.
\end{align*}}
After a simple algebraic manipulation in the equation above, we have
\begin{align*}
\left|\left(\rho_{1}\right)^{x,0}_{x',0}\right|\leq&\frac{1}{3}\sqrt{\left(\rho_{0}\right)^{2x,0}_{2x,0}}\left(\sqrt{\left(\rho_{0}\right)^{2x',0}_{2x',0}}+\sqrt{\left(\rho_{0}\right)^{2x',1}_{2x',1}}\right)\\
&\quad+\frac{1}{3}\sqrt{\left(\rho_{0}\right)^{2x,1}_{2x,1}}\left(\sqrt{\left(\rho_{0}\right)^{2x',0}_{2x',0}}+\sqrt{\left(\rho_{0}\right)^{2x',1}_{2x',1}}\right).
\end{align*}
Noticing that
{\footnotesize
\begin{equation}
\label{eq:ineqsqrt}
\sqrt{\left(\rho_{0}\right)^{2x,a}_{2x,a}}\left(\sqrt{\left(\rho_{0}\right)^{2x',0}_{2x',0}}+\sqrt{\left(\rho_{0}\right)^{2x',1}_{2x',1}}\right) \leq \sqrt{\left(\rho_{0}\right)^{2x',0}_{2x',0}}+\sqrt{\left(\rho_{0}\right)^{2x',1}_{2x',1}},
\end{equation}} 
since $\left(\rho_{0}\right)^{2x,a}_{2x,a}\leq1$ for $a=0,1$ we have that  
\begin{equation}
\label{eq:com_fac}
\left|\left(\rho_{1}\right)^{x,0}_{x',0}\right|\leq\frac{2}{3}\left(\sqrt{\left(\rho_{0}\right)^{2x',0}_{2x',0}}+\sqrt{\left(\rho_{0}\right)^{2x',1}_{2x',1}}\right).
\end{equation}
Now, from lemma \ref{lemma}
\begin{align*}
\sqrt{\left(\rho_{0}\right)^{2x',0}_{2x',0}}+\sqrt{\left(\rho_{0}\right)^{2x',1}_{2x',1}}\leq\sqrt{2},
\end{align*}
therefore, 
\begin{equation}
\label{eq:ineq}
\left|\left(\rho_{1}\right)^{x,0}_{x',0}\right|\leq\frac{2\sqrt{2}}{3}<1.
\end{equation}

\begin{theorem}
Let $\left(\rho_{L}\right)\in \mathbb{C}^{n/2^L}\times\mathbb{C}^{n/2^L}$ be the density matrix in the coarse level $L$, with $n\gg2^L$, where their off-diagonal elements in terms of the level 0 are given by equation (\ref{eq:cg0toLcoe}). Then, the absolute value of the coherence terms of $\left(\rho_{L}\right)$ is upper bounded by
\begin{equation}
    \left(\frac{2\sqrt{2}}{3}\right)^L,
\end{equation}
for $x\neq x'$ and $b,b'=0,1$.
\end{theorem}
\begin{proof}
From equation (\ref{eq:cg0toLcoe}), the following inequality can be written as
\begin{align}
   &\left|\left(\rho_{L}\right)_{x',b'}^{x,b}\right| \leq\frac{1}{3^L}\sum_{s,s'=0}^{2^{L-1}-1} \left( \left|\left(\rho_{0}\right)_{2^L x'+2^{L-1}b'\Delta x+s'\Delta x,0}^{2^L x+2^{L-1}b\Delta x+s\Delta x,0}\right|\right.\nonumber\\ 
   &\quad+ \left|\left(\rho_{0}\right)_{2^L x'+2^{L-1}b'\Delta x+s'\Delta x,0}^{2^L x+2^{L-1}b\Delta x+s\Delta x,1}\right| \left|\left(\rho_{0}\right)_{2^L x'+2^{L-1}b'\Delta x+s'\Delta x,1}^{2^L x+2^{L-1}b\Delta x+s\Delta x,0}\right|\nonumber\\
   &\quad+ \left.\left|\left(\rho_{0}\right)_{2^L x'+2^{L-1}b'\Delta x+s'\Delta x,1}^{2^L x+2^{L-1}b\Delta x+s\Delta x,1}\right|\right).
\end{align}
Now, using the inequality given by equation (\ref{eq:ineq_pop_coh})
\begin{align*}
    &\left|\left(\rho_{L}\right)_{x',b'}^{x,b}\right|\leq\frac{1}{3^L}\sum_{s=0}^{2^{L-1}-1}\nonumber\\
    &\quad\times\left[\sqrt{\left(\rho_{0}\right)_{2^L x+2^{L-1}b\Delta x+s\Delta x,0}^{2^L x+2^{L-1}b\Delta x+s\Delta x,0}}\sum_{s'=0}^{2^{L-1}-1}\left(\sqrt{\left(\rho_{0}\right)_{2^L x'+2^{L-1}b'\Delta x'+s'\Delta x,0}^{2^L x'+2^{L-1}b'\Delta x'+s\Delta x',0}}\right.\right.\nonumber\\
    &\quad+\left.\sqrt{\left(\rho_{0}\right)_{2^L x'+2^{L-1}b'\Delta x'+s'\Delta x,1}^{2^L x'+2^{L-1}b'\Delta x'+s\Delta x',1}}\right)\nonumber\\
    &\quad+\sqrt{\left(\rho_{0}\right)_{2^L x+2^{L-1}b\Delta x+s\Delta x,1}^{2^L x+2^{L-1}b\Delta x+s\Delta x,1}}\sum_{s'=0}^{2^{L-1}-1}\left(\sqrt{\left(\rho_{0}\right)_{2^L x'+2^{L-1}b'\Delta x'+s'\Delta x,0}^{2^L x'+2^{L-1}b'\Delta x'+s\Delta x',0}}\right.\\
    &\quad+\left.\left.\sqrt{\left(\rho_{0}\right)_{2^L x'+2^{L-1}b'\Delta x'+s'\Delta x,1}^{2^L x'+2^{L-1}b'\Delta x'+s\Delta x',1}}\right)\right].
\end{align*}
Now we can apply the inequality used in equation (\ref{eq:ineqsqrt}) to get
\begin{align*}
    \left|\left(\rho_{L}\right)_{x',b'}^{x,b}\right|&\leq\frac{2^{L-1}}{3^L}\left[\sum_{s'=0}^{2^{L-1}-1}\left(\sqrt{\left(\rho_{0}\right)_{2^L x'+2^{L-1}b'\Delta x'+s'\Delta x,0}^{2^L x'+2^{L-1}b'\Delta x'+s\Delta x',0}}\right.\right.\\
    &\quad+\left.\sqrt{\left(\rho_{0}\right)_{2^L x'+2^{L-1}b'\Delta x'+s'\Delta x,1}^{2^L x'+2^{L-1}b'\Delta x'+s\Delta x',1}}\right)\\
    &\quad+ \sum_{s'=0}^{2^{L-1}-1}\left(\sqrt{\left(\rho_{0}\right)_{2^L x'+2^{L-1}b'\Delta x'+s'\Delta x,0}^{2^L x'+2^{L-1}b'\Delta x'+s\Delta x',0}}\right.\nonumber\\
    &\quad+ \left.\left.\sqrt{\left(\rho_{0}\right)_{2^L x'+2^{L-1}b'\Delta x'+s'\Delta x,1}^{2^L x'+2^{L-1}b'\Delta x'+s\Delta x',1}}\right)\right],
    \end{align*}
    or
\begin{align}
    \left|\left(\rho_{L}\right)_{x',b'}^{x,b}\right|&\leq\frac{2^{L}}{3^L}\sum_{s'=0}^{2^{L-1}-1}\left(\sqrt{\left(\rho_{0}\right)_{2^L x'+2^{L-1}b'\Delta x'+s'\Delta x,0}^{2^L x'+2^{L-1}b'\Delta x'+s\Delta x',0}}\right.\nonumber\\
    &\quad+\left.\sqrt{\left(\rho_{0}\right)_{2^L x'+2^{L-1}b'\Delta x'+s'\Delta x,1}^{2^L x'+2^{L-1}b'\Delta x'+s\Delta x',1}}\right).
\end{align}
The proof can be conclude by using the lemma \ref{lemma} in the inequality above for $D=2^L$
\end{proof}

\end{document}